\documentclass[accepted]{uai2024}

\usepackage{natbib} % has a nice set of citation styles and commands
\usepackage{mathtools} % amsmath with fixes and additions
\usepackage{booktabs} % commands to create good-looking tables

\DeclareMathOperator{\E}{\mathbb{E}}
\usepackage[format=plain]{caption}
\usepackage{subcaption}
\usepackage[british]{babel}

\title{Spillover Detection for Donor Selection in Synthetic Control Models}

\usepackage{booktabs}
\usepackage{times}

\usepackage{algorithm}
\usepackage{algorithmicx}
\usepackage[noend]{algpseudocode}

\usepackage{cleveref}
\usepackage{amsmath,amsfonts,amssymb,amsthm,mathtools}
\usepackage{tikz}
\usetikzlibrary{shapes.geometric}
\usepackage{multicol}
\newtheorem{thm}{Theorem}[section]
\newtheorem*{thm*}{Theorem}
\newtheorem{define}{Definition}[section]
\author[1]{\href{mailto:<moriordan@spotify.com>}{Michael O'Riordan}{}}
\author[1,2]{\href{mailto:<ciaran.lee@ucl.ac.uk>}{Ciar\'{a}n M. Gilligan-Lee}{}}
\affil[1]{%
    Spotify
}
\affil[2]{%
    University College London
}

\begin{document}
\maketitle

\begin{abstract}%
Synthetic control (SC) models are widely used to estimate causal effects in settings with observational time-series data. To identify the causal effect on a target unit, SC requires the existence of correlated units that are not impacted by the intervention. Given one of these potential donor units, how can we decide whether it is in fact a \emph{valid} donor---that is, one not subject to spillover effects from the intervention? Such a decision typically requires appealing to strong \emph{a priori} domain knowledge specifying the units, which becomes infeasible in situations with large pools of potential donors. In this paper, we introduce a practical, theoretically-grounded donor selection procedure, aiming to weaken this domain knowledge requirement. Our main result is a Theorem that yields the assumptions required to identify donor values at post-intervention time points using only pre-intervention data. We show how this Theorem---and the assumptions underpinning it---can be turned into a practical method for detecting potential spillover effects and excluding invalid donors when constructing SCs. Importantly, we employ sensitivity analysis to formally bound the bias in our SC causal estimate in situations where an excluded donor was indeed valid, or where a selected donor was invalid. Using ideas from the proximal causal inference and instrumental variables literature, we show that the excluded donors can nevertheless be leveraged to further debias causal effect estimates. Finally, we illustrate our donor selection procedure on both simulated and real-world datasets.
\end{abstract}

%\begin{keywords}%
%  Synthetic Control, Sensitivity Analysis, Structural Causal Models
%\end{keywords}

\section{Introduction}

The ability to estimate causal effects is of fundamental importance in many domains, including medicine, economics, and industry 
\citep[see e.g.][]{lee2017causal,gilligan2020causing,richens2020improving, dhir2020integrating,perov2020multiverse,vlontzos2021estimating,gilligan2022leveraging, jeunen2022disentangling,reynaud2022d,van2023estimating}. 
In the absence of experimental data from randomised controlled trials or A/B tests, practitioners are faced with observational (non-randomised) data, and must rely on the assumptions, tools, and techniques of causal inference to estimate the impact of interventions. 
Synthetic control (SC) models, first introduced more than 20 years ago by \cite{abadie2003economic}, are widely used to estimate treatment effects in settings with observational \emph{time-series} (panel) data \citep[e.g.][]{abadie2010synthetic, abadie2015comparative, brodersen2015inferring, kreif2016examination}, and have been described by \cite{athey2017state} as ``arguably the most important innovation in the policy evaluation literature in the last 15 years''. 

To estimate the impact of an intervention on a target unit, SC requires time-series data for the target unit as well as time-series data for other units correlated with the target, often called donors. Crucially, the donors must not be impacted by the intervention. That is, there must no spillover effects from the intervention to the donors. The SC method uses pre-intervention data to construct a \emph{synthetic control unit} that matches the pre-intervention target as closely as possible.
The post-intervention evolution of this SC unit estimates the evolution of the target unit in a counterfactual world where the intervention did not occur, all else being equal. Therefore, the causal impact of the intervention can be estimated by comparing the observed, factual, post-intervention target to the counterfactual one.

Given time-series data for a potential donor, how can we determine that it is \emph{not} impacted by the intervention, and hence constitutes a \emph{valid} donor that can be used in the construction of SCs? Such a decision typically requires appealing to strong \emph{a priori} domain knowledge about the nature of the intervention and donors. Usually, we must already know (or assume) that the entire pool of potential donors is not subject to spillover effects from the intervention. However, in many real-world applications the pool of potential donors can be very large, such as when estimating the impact of a new feature on a large online platform \citep[][Section 5]{lin2023balancing}, and domain knowledge alone is unlikely to be adequate for donor selection due to the scale of the problem. 

In this work, we aim to relax the domain knowledge requirements by introducing a practical, theoretically-grounded donor selection procedure. This procedure can be used to augment partial knowledge about invalid donors, thereby allowing us to rely on a weaker form of domain knowledge than usually required to select valid donors.
We also relate the failure modes of our selection procedure to recent advances in the sensitivity analysis frameworks for SC \citep{zeitler2023non} and negative control \citep{miao2020confounding}, 
giving formal bounds on the bias for a given selection of donors, and further reducing the burden on domain knowledge to select a perfect set of valid donors. 

Our main result is a Theorem---based on techniques from proximal causal inference---that provides the assumptions required to identify and forecast values of specific donors at post-intervention time points using pre-intervention donor data only. This is in contrast to the standard SC identifiability, which additionally uses post-intervention donor values to predict the post-intervention counterfactual for the target unit.
%necessary for SC identifiability also facilitate using pre-intervention donor data to forecast the values of specific donors at post-intervention time points. 
The main assumptions required for a given donor's post-intervention values to be identified are that the donor is not impacted by the intervention, and the distribution of the underlying data generating mechanism in the past is representative of future data points. Therefore, if we use pre-intervention data to predict post-intervention data and get the wrong answer, either the underlying latent distribution has changed, the donor has been impacted by spillover effects from the intervention, or both. 

We use this result to detect potential spillover effects and exclude donors when constructing SCs. Importantly, we formally bound the potential bias introduced by this selection procedure due to false positives (\emph{excluding} donors \emph{not impacted} by spillover effects) and false negatives (\emph{including} donors \emph{impacted} by spillover effects) using sensitivity analysis. While the excluded donors aren't used in constructing the SC, we show how they can still be leveraged to further debias causal effect estimates in situations where the donors are noisy proxies of the latent dynamics \citep{shi2023theory}.

We conclude the paper by providing an empirical demonstration of our donor selection procedure on both simulated and real-world datasets.

In summary, our main contributions are as follows:
\begin{enumerate}
\item We prove a Theorem that yields the assumptions required to identify donor values at post-intervention time points using only pre-intervention data.
\item We introduce a practical donor selection procedure using this Theorem, that detects potential spillover effects and excludes invalid donors when constructing SCs.
\item We use sensitivity analysis to formally bound the bias in our SC causal estimate for a given selection of donors in situations where an excluded donor was indeed valid, or where a selected donor was invalid.
\item We provide a two-stage method that uses the excluded donors to further debias causal effect estimates.
\item We illustrate the performance of our donor selection procedure on both simulated and real-world datasets.
\end{enumerate}
    
\section{Related work}

\paragraph{\textbf{Identifiability of Synthetic Controls}} 
Historically, SC identifiability relied on assuming that the data generating process can be modelled as a latent linear factor model. With this assumption, the counterfactual is identified as a linear combination of valid donors. For instance, see \cite{abadie2003economic,abadie2010synthetic,abadie2015comparative}, and extensions of these approaches utilising Bayesian structural time-series by \cite{brodersen2015inferring}.

More recently, \cite{shi2022assumptions} argued that linearity emerges in a non-parametric manner if the target and donor units are in fact aggregations of ``smaller'' units (e.g. country-level data are aggregates of individual-level behaviours). 
The need for this ``aggregate unit'' assumption was subsequently removed by \cite{shi2023theory} and \cite{zeitler2023non}, who leveraged proximal causal inference to prove that the counterfactual can be non-parametrically identified as a (potentially non-linear) function of valid donors.

\paragraph{\textbf{Proximal Causal Inference}}
Proximal causal inference was initially investigated by \cite{kuroki2014measurement,miao2018identifying}, and has been further developed by \cite{tchetgen2020introduction}. It has been used, for instance, in long-term causal effect estimation by \cite{imbens2022long}, and formulated in terms of the graphical causal inference framework by \cite{shpitser2021proximal}. 
In the context of SC models, \cite{shi2023theory} and \cite{zeitler2023non} employed proximal causal inference to prove non-parametric identifiability.

Recently, \cite{liu2023proximal} proposed a novel SC model leveraging special donors called ``surrogates'', which are correlates of the causal effect itself, and can include potentially invalid donors.
Assuming the existence of such surrogates (as well as the requisite domain knowledge to identify them), \cite{liu2023proximal} demonstrated SC estimation based on post-intervention data alone.

\paragraph{\textbf{Sensitivity Analysis}} 
Later in the paper, we discuss scenarios where our donor selection procedure fails, and relate these to recent works on sensitivity analysis in SC \citep{zeitler2023non} and negative control \citep{miao2020confounding}. Sensitivity analysis in the causal inference literature has mainly been focused on investigating omitted variable bias in propensity-based models. This line of work originated in \cite{rosenbaum1983assessing} and \cite{imbens2003sensitivity}, with modern formulations provided by \cite{veitch2020sense} and \cite{cinelli2020making,cinelli2019sensitivity}. 

In the context of SC models, sensitivity analysis has been investigated by \cite{zeitler2023non} and \cite{nazaret2023misspecification}. The work by \cite{zeitler2023non} explored a formal framework for sensitivity to violations of identifiability of the full SC model, in particular, sensitivity to the existence of relevant latent variables with no observed proxies. \cite{nazaret2023misspecification}, on the other hand, explored mis-specifications to the standard linearity assumption in SC models. Both of these works assumed valid donors (i.e.\ no donors impacted by spillover effects from the intervention). 

\cite{miao2020confounding} investigated sensitivity analysis in the context of negative control \citep[widely used in epidemiological research, negative control variables can be viewed as specific proxy types within the proximal causal inference framework][]{tchetgen2020introduction,shi2022selective}. In particular, \citet{miao2020confounding} introduced \emph{positive} control outcomes (which are reminiscent of invalid donors in SC models), and discussed treating the (unknown) spillover effects as sensitivity parameters to evaluate the plausibility of causal estimates.

\section{Methods}

\subsection{Synthetic Control Structural Causal Model} \label{section: Pearl SCM}

\begin{figure*}[t]
\centering
\begin{subfigure}[b]{0.55\textwidth}
       \includegraphics[scale=0.7]{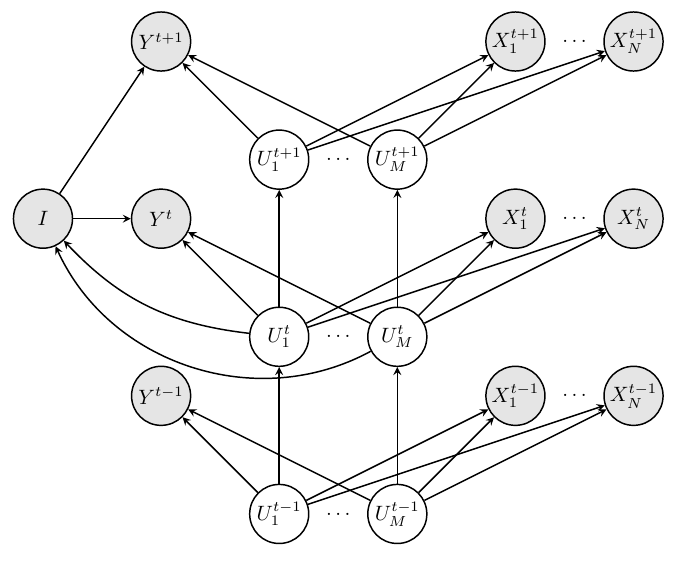}
       \caption{}
       \label{fig:syn_cont}
       \end{subfigure}
\begin{subfigure}[b]{0.4\textwidth}
    \includegraphics[scale=0.8]{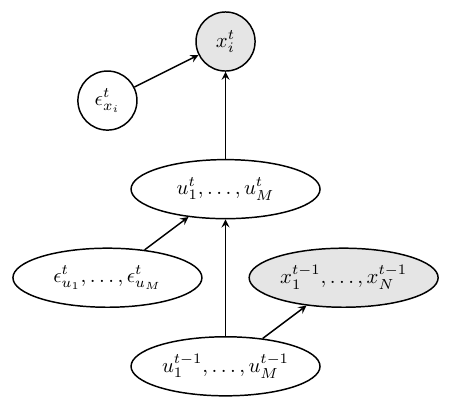}
    \caption{}
    \label{fig:forecast_donor_proxy}
    \end{subfigure}
       \caption{(a) SC DAG. Grey nodes are observed variables, white latent. The intervention $I$ is applied at time point $t$, and is taken to be $0$ for all time points before $t$, and $1$ for all time points from $t$ onwards. The noise terms for the target $Y$, donors $X$, and latents $U$ have been suppressed for ease of exposition. Note that autocorrelation in the target and donor time-series arises due to the causal links between the latents at different time points. As discussed in Appendix~\ref{section: proof}, the model and results can be readily generalised to the case where the evolution depends on $T$ time points. (b) Donor forecast. At time point $t-1$, the donors $x_1^{t-1},\dots,x_N^{t-1}$ are proxies for the latents $u_1^{t-1},\dots,u_M^{t-1}$, which allows us to write $x_i^t$ as a function of $x_1^{t-1},\dots,x_N^{t-1}$ and the noise terms at time $t$. The noise terms at time $t-1$ have been suppressed for ease of exposition.
       }
\end{figure*}

In Def.~\ref{definition: scscm} \citep[first provided by][]{zeitler2023non}, we formally define SC models in the structural causal model (SCM) framework of \cite{pearl2009causality}. 
    
    \begin{define}[Synthetic Control Structural Causal Model]\label{definition: scscm}
    SC structural causal models (SCSCMs) consist of \textbf{a set of $\mathbf{M}$ latent variables $\mathbf{U}$} and their distributions, \textbf{a set of $\mathbf{N}$ observed variables $\mathbf{X}$} representing the donor units, \textbf{a set of observed variables $\mathbf{Y, I}$} representing the target unit, and the intervention, and \textbf{a set of deterministic functions} mapping parents to their children in the causal structure in Figure~\ref{fig:syn_cont}, represented as a directed acyclic graph (DAG), each indexed by a specific time point $t$, such that
       \begin{enumerate}
           \item $u_i^t=m_i^t(u_i^{t-1}, \epsilon^t_{u_i})$ 
           \item $x^t_i=f_i^t(u_1^t,\dots,u_M^t, \epsilon^t_{x_i})$ 
           \item $y^t=g^t(u_1^t,\dots,u_M^t, I^t, \epsilon^t_{y})$ 
       \end{enumerate}
        where $\epsilon^t_{u_i}\sim P(\epsilon^t_{u_i})$,\, 
        $\epsilon^t_{x_i}\sim P(\epsilon^t_{x_i})$,\, and 
        $\epsilon^t_{y}\sim P(\epsilon^t_{y})$ are independent, exogenous error terms.  
    \end{define}
For simplicity, we sometimes follow \cite{zhang2022can} and suppress the functional dependence on the exogenous error terms for the latents $u_1,\dots,u_M$. We also often drop the indices on the donors and latents as follows: $x^t:=x_1^t,\dots,x_N^t$, $u^t:=u_1^t,\dots,u_M^t$, $\epsilon_x^t:=\epsilon_{x_1}^t,\dots,\epsilon_{x_N}^t$, and $\epsilon_u^t:=\epsilon_{u_1}^t,\dots,\epsilon_{u_M}^t$. 

This definition of SCSCMs generalises the standard latent linear factor model formulation of SCs \cite{abadie2010synthetic}. In particular, notice that the target $y^t$ and donors $x^t$ can be arbitrary functions of the latents $u^t$. We sometimes also use $y^t$ and $x^t$ to denote the \emph{values} these variables take, however, the difference will be clear from the context. Note that autocorrelation in the target and donor time-series arises due to the causal links between the latents at different time points. As discussed in Appendix~\ref{section: proof}, the definition of SCSCMs and our results can be readily generalised to the case where evolution depends on $T$ time points. 

We finish this section by defining the proxy variable completeness condition (Def.~\ref{definition: completeness}), and invariant causal mechanisms (Def.~\ref{definition: inv causal mech}), which are necessary for SC non-parametric identifiability, and form the basis for our Theorem~\ref{theorem: forecast}. 

The following completeness condition formally defines when a set of donors can be considered \emph{proxies} for a set of latent variables \citep{tchetgen2020introduction,zeitler2023non}.

    \begin{define}[Completeness Condition]\label{definition: completeness}
       For any square-integrable function $f$, if  $\E\left(f(x^t_1, \dots, x^t_N) \mid u^t_1, \dots, u^t_M \right) = 0,$ then $f(x^t_1, \dots, x^t_N)=0$ for any $t$.
    \end{define}
At a given time point, the completeness condition characterises how much ``information'' the donors have about the latent variables, in the sense that any variation in the $u$'s is captured by variation in $x$'s. For the rest of the paper, we assume that the set of donors $x^t_1,\dots, x^t_N$ can be treated as proxies for the latents $u^t_1,\dots, u^t_M$.

A \emph{causal mechanism} is the deterministic function that uniquely specifies a variable from its parents in the causal graph, and is equivalent to the conditional distribution of that variable given its (latent and observed) parents. We assume that causal mechanisms are \emph{invariant} (Def.~\ref{definition: inv causal mech}).

\begin{define}[Invariant Causal Mechanism]\label{definition: inv causal mech}
       A causal mechanism is invariant if it doesn't depend on time point $t$.
    \end{define}

The SCM formulation allows us to define (strong) interventions via the \emph{do-operator}, disconnecting the intervened variable from its parents in the causal graph, and assigning to it a specific value \citep{pearl2009causality}. Thus, we quantify the impact of an intervention $I$ on the target $y$ at time $t$ as
\begin{multline}
\label{eq:sc_att}
\tau=\underbrace{\E\left(y^t \mid \text{do}(I^t=1), I^t=1 \right)}_{\text{Observed}} \\-\,  
\underbrace{\E\left(y^t \mid \text{do}(I^t=0), I^t=1 \right)}_{\text{Counterfactual}}
\end{multline}
The first term is the observed, factual, post-intervention target, whereas the second term is the unobserved, counterfactual one. 
For SCSCMs, \cite{zeitler2023non} showed that if causal mechanisms are invariant, and if the donors are proxies for the latents, then 
the counterfactual is identified via a unique function $h$ of the (valid) donors such that $\E\left(y^t \mid \text{do}(I^t=0), I^t=1 \right)= \E\left(h(x^t, I^t=0) \right)$.

\subsection{Spillover Detection for Selecting Valid Donors}

The power of SC lies in the fact that the counterfactual in Eq.~\ref{eq:sc_att} can be identified even when there are post-intervention shifts in the exogenous errors $P(\epsilon_u^t)$ for the latents $u$.
However, to identify the counterfactual we still need to specify donor units $x$ that are \emph{valid}. A donor is valid if it adheres to the DAG depicted in Figure~\ref{fig:syn_cont}, and remains a proxy for the latents (Def.~\ref{definition: completeness}) at all time points. In particular, it must not be impacted by spillover effects from the intervention. Given an SCSCM from Def.~\ref{definition: scscm},  spillover effects manifest as post-intervention shifts in the donor errors $P(\epsilon_x^t)$.

Usually, one must appeal to strong \emph{a priori} domain knowledge in order to select valid donors.
Can we leverage data from the donor pool itself, augmenting incomplete domain knowledge, to gain confidence that a potential donor is valid?
Ideally, for a given donor candidate $x_i$, we would like to be able to use pre-intervention data to test for shifts in $P(\epsilon_{x_i}^t)$ that would rule out $x_i$ as a valid donor.
Such a procedure would allow us to select donors based on a weaker form of domain knowledge than that required to know exactly which donors adhere to the DAG.
This issue of selecting \emph{valid} donors is distinct from the usual considerations of selecting donors based on pre-treatment fit \citep[see e.g.,][]{ben-michael_2021}, and is in fact a prerequisite.

In Section~\ref{subsubsection: donor selection procedure}, we provide a Theorem showing that the assumptions necessary for SC non-parametric identifiability---the proxy variable completeness condition (Def.~\ref{definition: completeness}), and invariant causal mechanisms (Def.~\ref{definition: inv causal mech})---also facilitate forecasting donor values given additional assumptions.
We use this Theorem to introduce a practical, theoretically-grounded donor selection procedure based on detecting potential spillover effects to identify invalid donors.

\subsubsection{Donor Selection Procedure: Theory}
\label{subsubsection: donor selection procedure}
  
  \begin{thm}\label{theorem: forecast}
   If causal mechanisms are invariant, and the donors $x^{t-1}_1,\dots,x^{t-1}_N$ are proxies for the latents $u^{t-1}_1, \dots, u^{t-1}_M$ then, for each donor $x_i$, there exists a unique function $h_i$ such that for all time points $t$ we have: 
   \begin{equation}
   \label{eq:donor forecast}
   \E\left(x^t_i \right) = \E\left(h_i(x^{t-1}_1, \dots, x^{t-1}_N, P(\epsilon^t_{x_i},\epsilon^t_u)) \right)
   \end{equation}
  \end{thm}
The proof of this Theorem can be found in Appendix~\ref{section: proof}. 
Intuitively, because $x_i^t$ is a function of the latents at time point $t-1$, and because the donors at $t-1$ are proxies for those latents, we can swap the latents for the donors and write $x_i^t$ as a function of $x_1^{t-1},\dots,x_N^{t-1}$ and the exogenous error terms at time $t$ (see also Figure~\ref{fig:forecast_donor_proxy}).

We now use Theorem~\ref{theorem: forecast} to identify and exclude potentially invalid donors. Under conditions that a) donor $x_i$ is valid, and b) the exogenous error distributions $P\bigl(\epsilon_u^t\bigr)$ for the latents have not shifted at time $t$ relative to their pre-intervention values, then Theorem~\ref{theorem: forecast} implies that we can forecast post-intervention values for $x_i$ based on pre-intervention donor data alone. Conversely, failing to forecast $x_i^t$ implies the violation of condition a), b), or both. This observation forms the basis of our spillover detection method.

To use this theoretical result to get a practical method for flagging invalid donors, we assume that forecast errors are due to the donor unit error terms (due to being impacted by the intervention) and not the latent error terms. Crucially, we assume that failing to forecast $x_i^t$ means that condition a) is violated---that $P\bigl(\epsilon_{x_i}^t\bigr)$ has shifted\footnote{For example, $\epsilon_{x_i}^{t-1}\sim\mathcal{N}\bigl(0, \sigma_{x_i}\bigr)$, and $\epsilon_{x_i}^t\sim\mathcal{N}\bigl(\tau_{x_i},\sigma_{x_i}\bigr)$, where $\tau_{x_i}$ is the spillover effect on donor $x_i$.} relative to pre-intervention values---and $x_i$ is not a valid donor. 

In general, without additional domain knowledge, we can’t be sure that the unpredictability of $x_i$ is due to spillover effects and not changes in $u_i$. However, in Section~\ref{sensitivity analysis sec} we employ sensitivity analysis to formally bound the bias in our SC causal estimate violations to this assumption – ie. false positives, when we exclude a donor $x_i$ even though it was the latent $u_i$ that changed, and false negatives, when we include a donor $x_i$ even though it was impacted by spillover effects. In particular, when our spillover detection incorrectly flags a donor $x_i$ (which is a false positive), then excluding this donor will only introduce bias if the remaining selected donors don’t satisfy the completeness condition (i.e. there is omitted variable bias). This sensitivity analysis is the bridge that brings us from the theory of Theorem~\ref{theorem: forecast} to our practical donor selection method in Algorithm~\ref{alg}. 

As we demonstrate in Section~\ref{sec: experiments}, violations of condition b)---shifts in the latents---are only an issue for the selection method if they occur at the same time point as the donor forecast. If the latents shift later in the post-intervention period, this does not bias donor selection, as this later post-intervention data is not used in our selection procedure. Lags between the intervention and spillover effects can be dealt with by forecasting on coarsened, time-averaged donor data\footnote{For example, weekly averages instead of daily. Note that the time averaging should not combine pre- and post-intervention data.}. Furthermore, time averaging can reduce false negatives in cases with very noisy donors (see Figure~\ref{fig:spillover_time_averaged}). However, the longer time windows also increase the risk of false positives due to potential shifts in the latent distributions.

For the rest of the paper, we restrict our attention to \emph{linear} SC models, both because of their ubiquity, and also, more practically, to be able to discuss concrete sensitivity analysis bounds \citep{miao2020confounding,zeitler2023non}.
We make no distinction, however, between linear SCs as a consequence of the target and donors being linear functions of the latents, and linearity emerging non-parametrically due to the ``aggregate unit'' assumption of \cite{shi2022assumptions}.
For demonstration purposes, we also focus on simple linear models for estimating $h_i$ as part of the donor selection procedure, although this is \emph{not} a strict requirement (even when restricting to linear SC models, and linear sensitivity analysis). In principle, we can swap the linear model in this procedure for any flexible machine-learning model, using e.g.\ conformal inference for constructing calibrated prediction intervals, perhaps improving performance \citep[see][for applications of conformal inference in the context of SC models]{chernozhukov2021exact, shi2023theory}.
%However, it should be noted that---following the assumptions of Theorem~\ref{theorem: forecast}---linear models are likely a reasonable choice for the spillover detection procedure in scenarios where linear SC models identify the causal effect.

\subsubsection{Donor Selection Procedure: Practical Method}
\label{subsubsection: practical donor selection procedure}

In Algorithm~\ref{alg}, we present pseudocode for our spillover detection procedure.
The normalisation step is important for ensuring that the procedure remains invariant to the scale of the donors. 

\begin{algorithm}[ht]
\caption{Spillover Detection For Candidate Donor $x_i$}\label{alg}
%\vspace{10pt}
\textbf{Inputs:}\\ 
\begin{itemize*}
\item Training data and labels $\{x_1^{t'-1},\dots,x_N^{t'-1};\,\,x_i^{t'}\}$ for all pre-intervention time points $1<t'<t$ \\
\item Test data and label $\{x_1^{t-1},\dots,x_N^{t-1};\,\,x_i^t\}$\\
\item Posterior predictive interval (PPI) bound $\phi$ (e.g.\ 80\%)\\
\end{itemize*}
\textbf{Outputs:}\\
\begin{itemize*}
\item Procedure $S1$: Prediction absolute error $\bigl|x^t_i-\widehat{x}^t_i\bigr|$\\
\item Procedure $S2$: $0$ if $x^t_i$ is inside $\phi$ PPI, else $1$
\end{itemize*}
\begin{algorithmic}[1]
\State Normalise the data and labels
\State Regress $x_i^{t'}$ on $x_1^{t'-1},\dots, x_N^{t'-1}$ $\,\forall\,t'<t$ to obtain $\widehat{h}_i$
\State Predict $\widehat{x}_i^t,\,[\widehat{x}_{i,-}^t,\, \widehat{x}_{i,+}^t]=\widehat{h}_i\left(x_1^{t-1},\dots,x_N^{t-1};\,\phi\right)$
\State Set $A=\bigl|x^t_i-\widehat{x}^t_i\bigr|$
\State Set $B=0$ if $\widehat{x}_{i,-}^t<x_i^t<\widehat{x}_{i,+}^t$, else $B=1$
\State Output $A$ for selection procedure $S1$
\State Output $B$ for selection procedure $S2$
\end{algorithmic}
\end{algorithm}

We assume the following linear model for the forecast
\begin{equation}
\label{eq:linear donor forecast}
    x_i^{t'}\sim\mathcal{N}\left(\alpha_i+\beta_{ij}\,x_j^{t'-1},\,\sigma_{x_i}\right)
\end{equation}
where $1 < t'<t$, the index $j$ runs from $1$ to $N$, 
and we use the Einstein summation convention such that a repeated index implies summation $\beta_{ij}\,x_j^{t'-1} := \sum_j \beta_{ij}\,x_j^{t'-1}$.
Note that the coefficients $\alpha_i$ and $\beta_{ij}$ are the same for different time points, encoding the fact that $h_i$ should be independent of time.
In Section~\ref{subsection: simulated data}, we demonstrate the following two approaches, $S1$ and $S2$, for selecting donors based on Algorithm~\ref{alg}, using the regression model in Eq.~\ref{eq:linear donor forecast}:
\begin{itemize}
\item $S1$: Select donors with the smallest difference between their actual and predicted values at the time of intervention, i.e.\ $\min_{x_i}\bigl(\bigl|x^t_i-\beta_{ij}\,x_j^{t-1}\bigr|\bigr)$.
\item $S2$: Select donors with values $x^t_i$ falling within some specified posterior predictive intervals (e.g.\ 80\%).
\end{itemize}

\subsection{Using Excluded Donors To Debias Linear Synthetic Control Models}\label{subsection: linear sc models}
While the donors excluded by Algorithm~\ref{alg} aren’t used in constructing the SC counterfactual, they can still be used to debias SC causal effect estimates, as we now show. The experiments in Section~\ref{sec: experiments} use the following linear model for constructing SCs based on pre-intervention time points.
\begin{equation}
\label{eq: linear sc model}
y^t \sim \mathcal{N}\left(\alpha+\beta_i\,x_i^t\,,\,\sigma_y\right)
\end{equation}
As discussed by \cite{shi2023theory}, in situations where the donors are noisy (imperfect) proxies of the latents, the SC model in Eq.~\ref{eq: linear sc model} might fail to estimate consistent donor weights $\beta_i$. Using proximal causal inference \citep{tchetgen2020introduction,shi2023theory}, we can leverage proxies not used in the construction of the SC to debias the estimates. In the proximal causal inference literature, this debiasing is typically accomplished with two-stage least squares estimation, or the generalised method of moments.
For the experiment in Figure~\ref{fig:proxy_debiased}, we opt to jointly model the target and donors as
\begin{equation}
\label{eq:mv_normal_proxies}
\begin{pmatrix}y^t\\ x^t_1 \\ \vdots \\ x^t_N\end{pmatrix}\sim  
\mathcal{N}\left(
\begin{pmatrix}
\alpha+\beta_i\,x^t_i\\ 
\gamma_1+\lambda_{1i}\,z^t_i \\ \vdots \\ \gamma_N+\lambda_{Ni}\,z^t_i
\end{pmatrix},\, \Sigma
\right)
\end{equation}
where $z_i$ are the (potentially invalid) donors excluded by our selection procedure. 
Such a model can be readily estimated using a probabilistic programming language.

How is it possible that we can use excluded donors $z_i$, potentially impacted by spillover effects from the intervention, when constructing debiased SC estimates? The key point is that we only ever use pre-intervention data from these excluded donors, and so the SC estimates are unaffected by the post-intervention dynamics of the donors $z_i$. 

To build intuition about the above procedure, consider the following model. $Y$ is analogous to the target variable, $U$ the latent, and $X$ the donor (ie. a noisy proxy of the latent $U$). We would like to estimate the donor weight $\gamma$. Because $U$ is latent, we can’t estimate the expectation on the right hand side. If we regress $Y$ on $X$ we won’t get a consistent estimate of $\gamma$ because the errors are correlated with $X$. 

$$
\begin{aligned}
&\mathbb{E}(Y|U) =\alpha\,U, \quad \mathbb{E}(X|U) =\beta\,U, \\ 
&\implies \mathbb{E}(Y|U) =\gamma\,\mathbb{E}(X|U)
\end{aligned}
$$ 

Instead, if we have an additional proxy variable $Z$, we can rewrite the model as

$$
\begin{aligned}
&\mathbb{E}(Y|Z) =\alpha\,\mathbb{E}(U|Z), \quad  \mathbb{E}(X|Z) =\beta\,\mathbb{E}(U|Z), \\ &\implies \mathbb{E}(Y|Z) =\gamma\, \mathbb{E}(X|Z)
\end{aligned}$$

In this case, $Z$ is observed and so we can estimate the expectation on the right hand side, and then get a consistent estimate for the donor weight $\gamma$. One approach for estimation (often used for estimating instrumental variables) is a two-stage estimator where we first regress $X$ on $Z$ to estimate $\hat{X}$ (sometimes referred to as a proximal control variable), and then regress $Y$ on $\hat{X}$. The multivariate model in equation 5 above effectively combines this two-stage process into a single model 
\citep[similar models can be used for estimating instrumental variables, see e.g.][Section 23.4]{gelman_hill_2006}.

\begin{figure*}[t]
    \centering
    \includegraphics[scale=0.55]{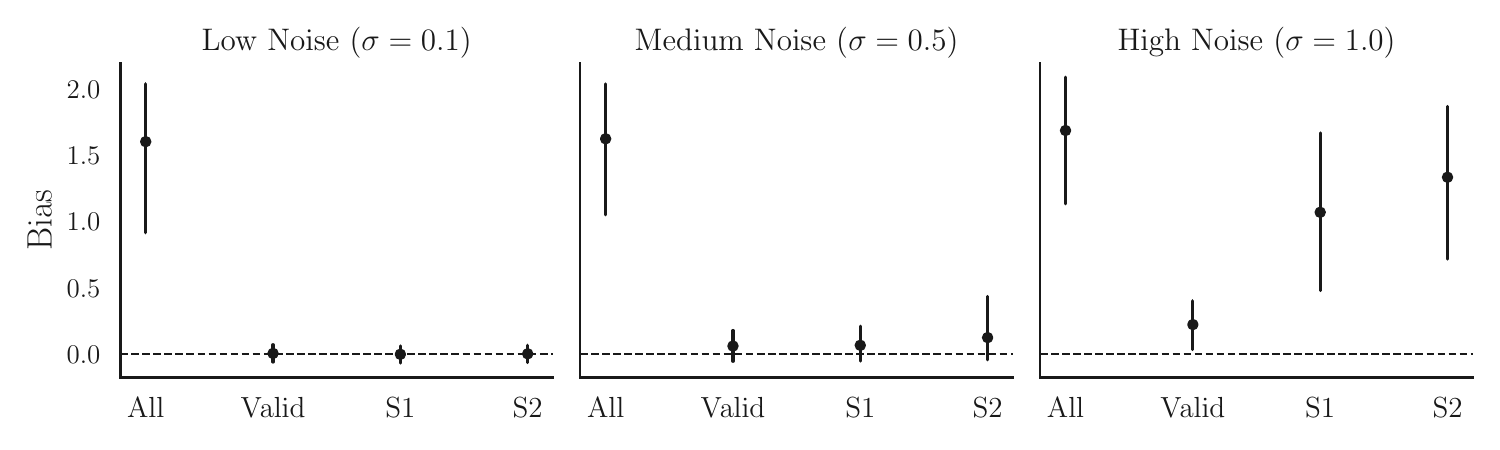}
    \caption{Bias $\widehat{\tau}-\tau$ and 95\% Monte Carlo confidence intervals
    for 2000 simulated datasets. The data generating process is described in Appendix~\ref{section: sim data}. 
    To construct SCs, we must identify the subset of donors that are valid. The horizontal axis shows the procedure used to identify this potentially valid set. 
    For comparison, the case labelled \emph{All} shows the expected bias of $1.6$ when we assume that all donors are valid, and the case labelled \emph{Valid} shows the bias when we have perfect knowledge about which donors are valid.
    Our $S1$ and $S2$ donor selection procedures are described in Section~\ref{subsubsection: donor selection procedure}, and Algorithm~\ref{alg}.
    The standard deviation of the donor noise term, $\epsilon_x^t\sim\mathcal{N}\left(0, \sigma\right)$, increases from left to right. As the noise increases, the spillover detection procedure is more likely to return false negatives, which increases the bias due to invalid donors. Note that even with optimal selection of valid donors, the estimates can still be biased due to donor noise.
    }
    \label{fig:sims_low_med_high_noise}
\end{figure*}

\begin{figure}[t]
    \centering
    \includegraphics[width=\linewidth]{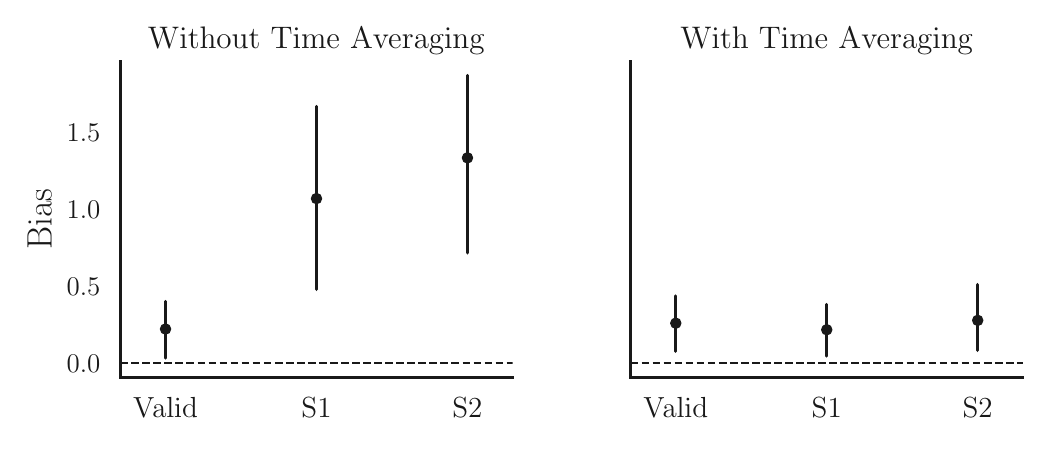}
    \caption{The left panel shows the bias in the high noise case of Figure~\ref{fig:sims_low_med_high_noise}. 
    Our selection procedures $S1$ and $S2$ are severely biased due to invalid donors. 
    The right panel shows where we time average the donor data, in buckets of 5 time points, before passing through the spillover detection. Note that we do not average pre- and post-intervention data in the same bucket, and we use the original, non-averaged data in the SC model. The averaging reduces false negatives such that the performance is on par with the optimal selection of valid donors. We address the residual bias in Figure~\ref{fig:proxy_debiased}.
    }
    \label{fig:spillover_time_averaged}
\end{figure}

\begin{figure}[t]
    \centering
    \includegraphics[width=\linewidth]{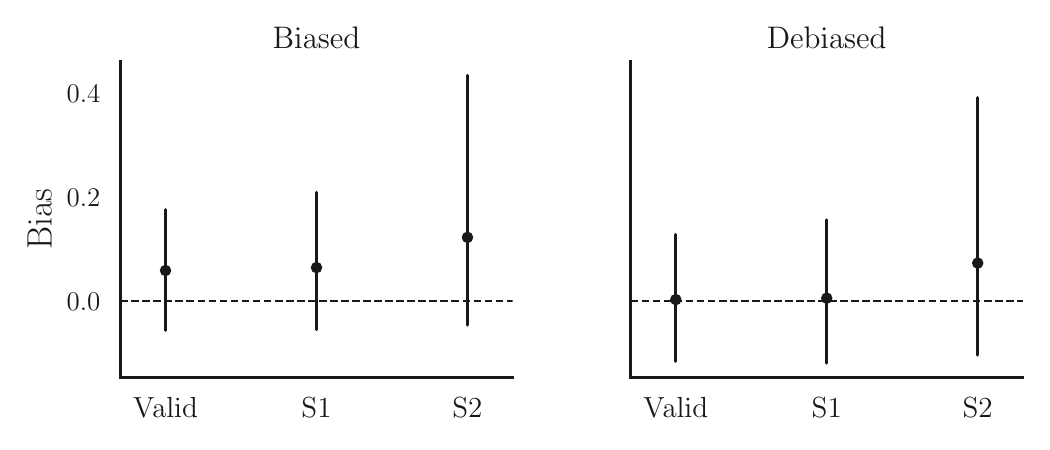}
    \caption{The left panel shows the bias in the medium noise case of Figure~\ref{fig:sims_low_med_high_noise}. 
    The estimated causal effects are biased even with optimal selection of valid donors. 
    This is because the donors are noisy (imperfect) proxies of the latents, and so conditioning fails to fully close backdoor paths.
    As described in Section~\ref{subsection: linear sc models}, we can debias these estimates using proximal causal inference to leverage donors excluded by the selection procedure.
    The right panel shows that the \emph{Valid} and $S1$ selection procedures now give unbiased estimates of the causal effect, although there is still some bias with procedure $S2$ due to invalid donors. As in Figure~\ref{fig:spillover_time_averaged}, time averaging would reduce the bias in $S2$.
    }
    \label{fig:proxy_debiased}
\end{figure}

\begin{figure}[t]
    \centering
    \includegraphics[width=\linewidth]{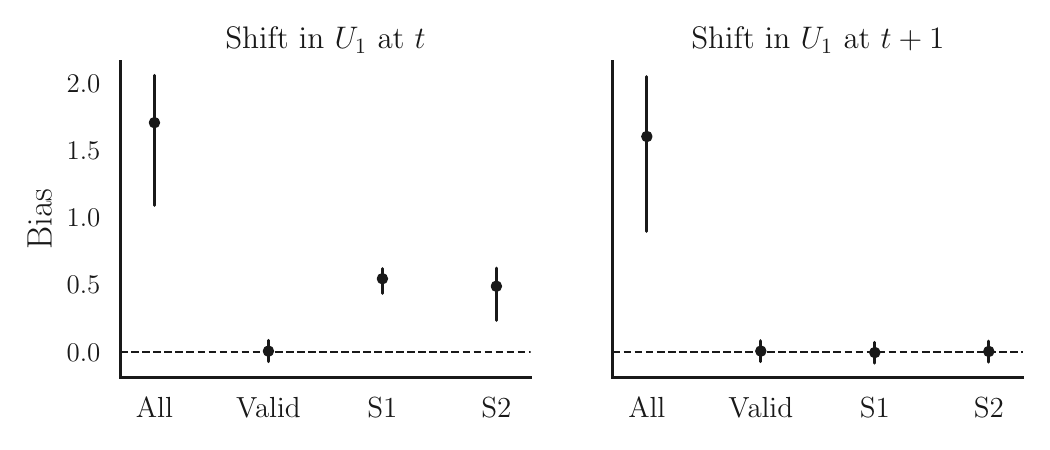}
    \caption{The left panel shows the bias when the error term for one of the latents shifts \emph{at the same time} as the intervention. The error term for $U_1$ shifts from $\epsilon_{u_1}\sim\mathcal{N}\left(0,1\right)$ pre-intervention, to $\epsilon_{u_1}\sim\mathcal{N}\left(0.5, 1\right)$ post-intervention. Our selection procedure incorrectly flags this as a spillover effect (false positive), thereby excluding donors that depend on $U_1$. The right panel shows where this latent shift occurs just \emph{after} the intervention. In this case our donor selection procedure recovers an unbiased estimate of the causal effect.
    }
    \label{fig:u_shift}
\end{figure}

%\begin{figure*}[t]
%\centering
%\begin{subfigure}[b]{0.45\textwidth}
%       \includegraphics[scale=0.58]{germany.pdf}
%       \caption{}
%       \label{fig:germany}
%\end{subfigure}
%\begin{subfigure}[b]{0.45\textwidth}
%       \includegraphics[scale=0.58]{prop99.pdf}
%       \caption{}
%       \label{fig:prop99}
%\end{subfigure}
%       \caption{(a) Estimated causal effect of German reunification. The results are very similar to Figure 3 from \cite{abadie2015comparative}, with a slightly larger positive initial ``demand boom'', and a flatter trend after the year 2000. 
%       (b) Estimated causal effect of the California tobacco tax. The results are very similar to Figure 3 from \cite{abadie2010synthetic}. Interestingly, Utah is excluded from the SC due to a relatively large decrease in sales at the time of the intervention, with New Mexico instead getting a non-zero weight. 
%       }
%\end{figure*}

\subsection{Sensitivity Analysis}\label{sensitivity analysis sec}

In this Section, we discuss the false positive (excluding donors not impacted by spillover effects), and false negative (including donors impacted by spillover effects) failure modes of our donor selection procedure. \cite{zeitler2023non} provide a general framework for sensitivity analysis in non-parametric SC models, bounding the bias when there are relevant latent variables with no observed donors as proxies (omitted variable bias). They also give a formula for calculating these bounds in the case of linear SC models, which we make use of in Sections~\ref{subsubsection: ov bias}, and \ref{subsubsection: fp bias}. In the context of negative control, \citet{miao2020confounding} introduce positive control outcome variables (which are similar to invalid donors), and discuss treating the spillover effect as a sensitivity parameter to investigate the plausibility of causal effect estimates. We make use of this approach in Section~\ref{subsubsection: fn bias}.

\subsubsection{Relevant Latents With No Observed Donors}
\label{subsubsection: ov bias}

Firstly, we discuss omitted variable bias due to the absence of observed donors. When there are latent confounding variables with no observed donors to act as proxies, we cannot close all backdoor paths and so our estimated causal effect will be biased. For linear SC models, this bias can be bounded using Eq.\ 3 from \cite{zeitler2023non}. In particular, let $x_i$ be the observed donors selected to construct the SC, and $\beta_{x_i}$ the corresponding donor weights. The potential omitted variable bias due to unobserved donors is then
\begin{multline}
\label{eq:sensitivity_analysis_unobserved}
\text{OV Bias} \leq N \times\max_{x_i}\left(\left|\beta_{x_i}\right|\right)\\
\times\max_{x_i}\left(\left|\E\left(x^\text{pre}_i\right)-
\E\left(x^\text{post}_i\right)\right|\right)
\end{multline}
Note that for this to be a \emph{valid} bound, we must assume that the observed donors are at least as important as the unobserved ones \citep[in the sense that the maximum weight and post-intervention shift for the observed donors is larger than that of the unobserved donors][]{zeitler2023non}. This is similar to assumptions used in propensity-based sensitivity analysis, where to obtain valid bounds one requires that unobserved confounders are at most as important to estimation as the observed confounders \citep{veitch2020sense}.

\subsubsection{Relevant Latents With No Selected Donors}
\label{subsubsection: fp bias}

A false positive result from the spillover detection procedure potentially introduces a related form of omitted variable bias. In particular, if \emph{all} donors acting as proxies for a relevant latent variable are excluded by the selection procedure, then the resulting SC causal estimate will be biased, in a similar manner to Section~\ref{subsubsection: ov bias}. However, in this case we can bound the bias using weaker assumptions than necessary for Eq.~\ref{eq:sensitivity_analysis_unobserved}, because the excluded donors are actually observed. Let $x_i$ be the donors selected to construct the SC, and $z_j$ be the donors excluded by the spillover detection procedure. The potential omitted variable bias due to false positives is then
\begin{multline}
\label{eq:sensitivity_analysis_excluded}
\text{FP Bias} \leq N \times\max_{x_i}\left(\left|\beta_{x_i}\right|\right)\\
\times\max_{z_j}\left(\left|\E\left(z^\text{pre}_j\right)-
\E\left(z^\text{post}_j\right)\right|\right)
\end{multline}

\subsubsection{Selected Donors Impacted By The Intervention}
\label{subsubsection: fn bias}

Finally, we discuss the potential bias introduced by false negatives. The approach taken here is similar to Section 6 of \cite{miao2020confounding}.
Let $x_i$ be the donors selected to construct the SC, and $\tau_{x_i}$ the corresponding spillover effects from the intervention. 
The potential bias due to false negatives is then
\begin{equation}
\label{eq:sensitivity_analysis_selected}
\text{FN Bias} \leq N \times\max_{x_i}\left(\left|\beta_{x_i}\right|\right)
\times\max_{x_i}\left(\left|\tau_{x_i}\right|\right)
\end{equation}
We don't know $\tau_{x_i}$, but \citep[following][]{miao2020confounding} we can treat it as a sensitivity parameter to gauge the plausibility of our causal effect estimate.
For example, if we have domain knowledge bounding $\tau_{x_i}$ then this can be used in Eq.~\ref{eq:sensitivity_analysis_selected} to bound the bias. We can also judge how large the spillover effect would have to be in order for the estimated causal effect to change sign.

\section{Experiments}
\label{sec: experiments}
In this Section, we illustrate the performance of our donor selection procedure in different scenarios using simulated data. We also demonstrate the procedure on real-world data by applying it to semi-synthetic variants of the German Reunification \citep{abadie2015comparative}, and California Tobacco Control \citep{abadie2010synthetic} datasets.

\subsection{Simulated Data}
\label{subsection: simulated data}

We construct 2000 simulated datasets according to the data generating process described in Appendix~\ref{section: sim data}, with large pools of potential donors, most of which are impacted by spillover effects from the intervention. 
Each dataset consists of a target timeseries, 10 latents, and a pool of 1000 potential donors. In each pool, a random set of $80\%$ of the donors are invalid, impacted by a spillover effect of $-2$.

For a given dataset, we must first attempt to identify the subset of valid donors. 
We refer to the set of donors returned by a selection procedure as the \emph{potentially valid donors} (PVDs).
Next, in order to simplify comparisons between different selection procedures, we estimate the SC using a sample of 10 donors from the PVDs. 
This sampling step ensures that the resulting SC models all have the same number of parameters, regardless of the number of PVDs identified by the different selection procedures\footnote{
For example, if procedure $S1$ returns 10 PVDs, and procedure $S2$ returns 50 PVDs, we sample 10 donors from $S2$'s PVDs such that the final SC models both have 10 donors.
}.
In Appendix~\ref{section: sparse weights}, we compare this approach to one where we estimate SCs using the full set of PVDs, and employ regularization to enforce sparsity in the donor weights. 
For the data generating process considered here, the expected bias is the same with both approaches. We use the sampling approach to focus the comparisons on the selection of \emph{valid} donors, rather than also having to consider possible differences due to selecting donors based on pre-treatment fit.

In Figure~\ref{fig:sims_low_med_high_noise}, we show the bias after applying our donor selection procedures $S1$\footnote{
For selection procedure $S1$, we select the 10 donors with the smallest forecast errors, and so the sampling step discussed previously is not relevant for this procedure.
}
and $S2$. 
For comparison, the case labelled \emph{All} shows the bias if we assume that all the donors are valid, and the case labelled \emph{Valid} shows the optimal scenario if we had perfect knowledge about which donors are valid. Our selection procedures are very close to optimal in both the low and medium donor noise cases. However, the performance degrades as the donor noise becomes comparable in magnitude to the spillover effect. In Figure~\ref{fig:spillover_time_averaged}, we address this issue and show how time averaging improves performance when the donors are very noisy. In Figure~\ref{fig:proxy_debiased}, we leverage excluded donors to further debias effect estimates in situations with noisy donors, as discussed in Section~\ref{subsection: linear sc models}. Finally, in Figure~\ref{fig:u_shift}, we show how contemporaneous shifts in the latent distributions can bias our selection procedure, and that these latent shifts are not an issue if they occur later in the post-intervention period.

\begin{figure*}[t]
\centering
\begin{subfigure}[b]{0.45\textwidth}
       \includegraphics[width=0.85\linewidth]{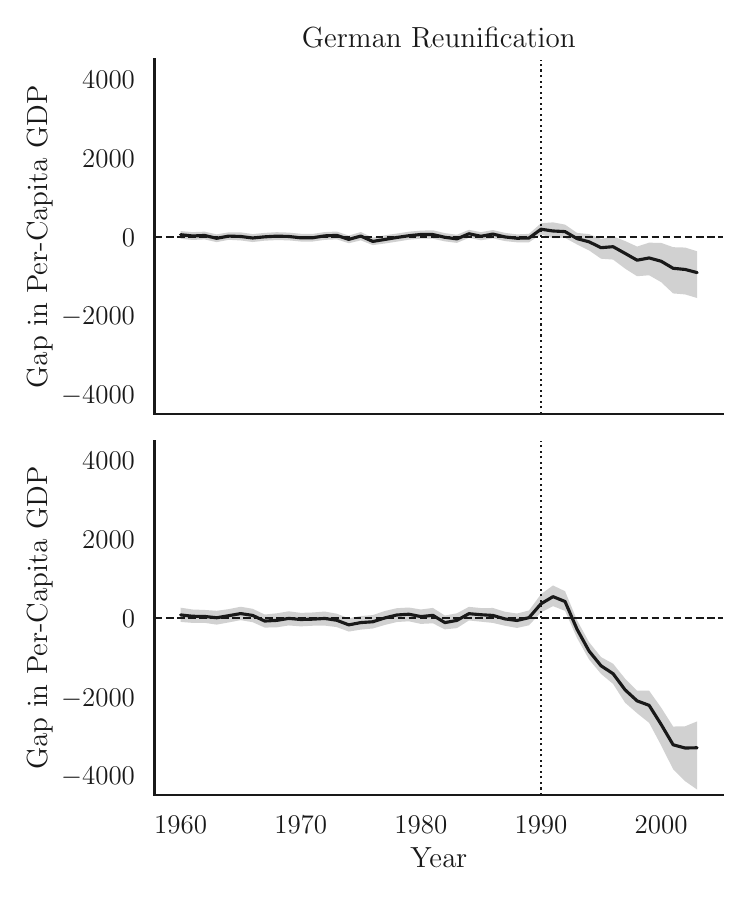}
       \caption{}
       \label{fig:germany}
\end{subfigure}
\begin{subfigure}[b]{0.45\textwidth}
       \includegraphics[width=0.85\linewidth]{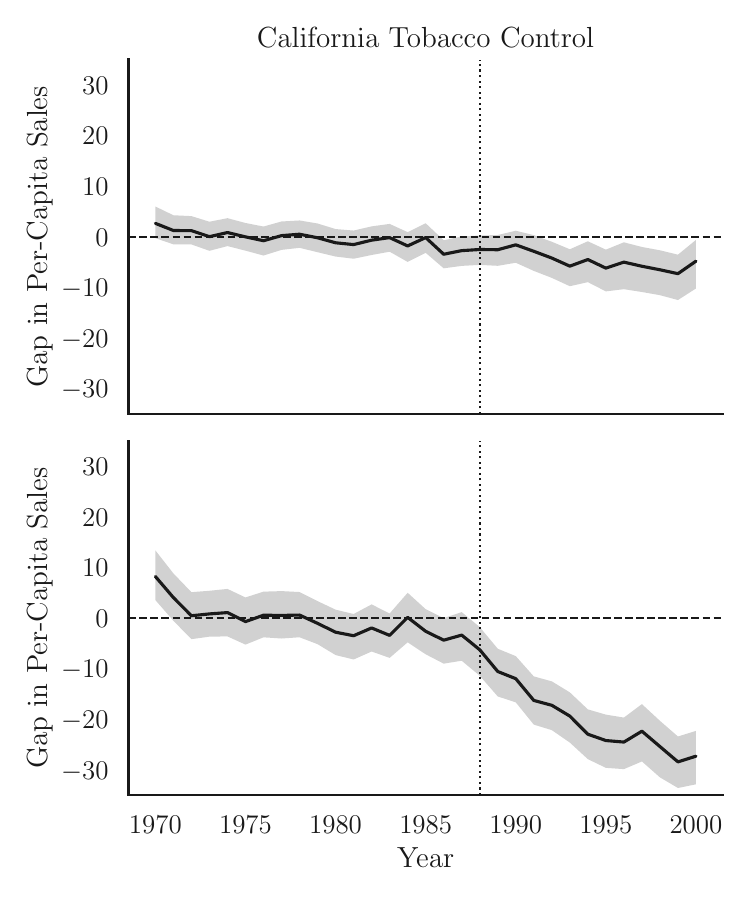}
       \caption{}
       \label{fig:prop99}
\end{subfigure}
       \caption{(a) Estimated causal effect of German reunification. The intervention time is indicated by the vertical dotted line, and the shaded area gives the 95\% uncertainty interval. In the top panel, the pool of potential donors includes a semi-synthetic unit that is a noisy proxy of West Germany. This invalid donor receives a large weight in the SC model, which biases the estimate towards zero.
       The bottom panel shows the effect estimated after applying our donor selection procedure. The invalid donor is correctly flagged and excluded from the SC model. In this case, the results are very similar to Figure 3 from \cite{abadie2015comparative}.
       (b) Estimated causal effect of the California tobacco tax. 
       The top panel includes a semi-synthetic unit that is a noisy proxy of California, biasing the estimate towards zero. In the bottom panel, this invalid donor has been excluded by our selection procedure, and the resulting 
       effect estimate is very similar to Figure 3 from \cite{abadie2010synthetic}. 
       }
\end{figure*}

\subsection{Semi-Synthetic Data}

In this Section, we further validate our donor selection procedure with using real-world data. In particular, we consider semi-synthetic variants of the 1990 reunification of West and East Germany \citep{abadie2015comparative}, and the 25 cents tobacco tax increase in California in 1988 \citep{abadie2010synthetic}. 
For each dataset, we introduce a semi-synthetic unit to the pool of potential donors, constructed to be a noisy proxy of the target as $x_\text{syn}^t\sim\mathcal{N}\left(y^t, \sigma\right)$.
Being predictive of the targets, these invalid donors receive large weights in the SC models and bias the effect estimates towards zero.
This also highlights the critical distinction between selecting \emph{valid} donors, and the usual considerations of selecting donors based on pre-treatment fit. 

In Figure~\ref{fig:germany}, we show the estimates for the effect of German reunification on Germany's per-capita GDP, and in Figure~\ref{fig:prop99}, the estimates for the effect of California's tobacco tax increase on per-capita pack sales. 
The SC models in the top panels include the semi-synthetic, invalid donors, which results in causal effect estimates much closer to zero than the original findings of \cite{abadie2010synthetic,abadie2015comparative}. The bottom panels show the effect estimates after applying our donor selection procedure. In this case, the semi-synthetic, invalid donors are correctly flagged and excluded from the SC models, resulting in causal effect estimates that are consistent with \cite{abadie2010synthetic,abadie2015comparative}.

\section{Conclusion}

In this paper, we presented a practical, theoretically-grounded donor selection procedure for SC models, aimed at weakening the domain knowledge requirements for selecting valid donors. 
This procedure augments partial knowledge about invalid donors, thereby reducing the burden on the practitioner to explicitly know that a (potentially very large) pool of donors is not impacted by spillover effects from the intervention.
Working in the structural causal model framework, we utilised techniques from proximal causal inference to show that the assumptions necessary for SC identifiability also facilitate forecasting post-intervention donor values from pre-intervention data. We used this result to detect potential spillover effects, and exclude invalid donors when constructing SCs.
Furthermore, in the context of recent works on sensitivity analysis, we discussed bounding the bias due to false positive and false negative selection errors.
We concluded by providing an empirical demonstration of our selection procedure on both simulated and real-world datasets.

\bibliography{main}

%\section{Technical conditions for proof}
%The proof of Theorem~\ref{theorem: forecast} depends on certain technical conditions, which we now briefly discuss. See Appendix C of \cite{shi2021proximal} for full details. To prove existence of the function $H^t$, consider the space of all square-integrable functions $s$, denoted $L^2\{F(s)\}$, with respect to a cumulative distribution function $F(s)$. This is a Hilbert space with inner product given by $\langle f,g \rangle = \int f(s)g(s)dF(s)$. Denote by $K_x$ the conditional expectation operator 
%$L^2\{F(w|x)\}\rightarrow L^2\{F(\lambda |x)\}$, with $K_xh = E[H(w) | \lambda, x]$ for $H \in L^2\{F(w|x)\}$, and denote by $ (\tau_{x,n}, \varphi_{x,n}, \psi_{x,n})_{n=1}^\infty $ a singular value decomposition of $K_x$. Given the following regularity conditions:
%\begin{enumerate}
%    \item $\iint f(w | \lambda, x)f(\lambda | w, x)dwd\lambda < \infty$,
%    \item $\int f^2(y | \lambda, x)f(\lambda | x)d\lambda < \infty$,
%    \item $\sum_{n=1}^\infty| \langle f(y | \lambda, x), \psi_x,n \rangle |^2 < \infty$,
%\end{enumerate}
%Picard’s theorem implies the existence of the required function $H^t$ in Theorem~\ref{theorem: forecast}.

\newpage

\onecolumn

%\title{Spillover Detection for Donor Selection in Synthetic Control Models\\(Supplementary Material)}
%\maketitle

\section*{Appendix}
\appendix

\section{Proof of Theorem~\ref{theorem: forecast}}\label{section: proof}

 \begin{thm*}
   If causal mechanisms are invariant, and the donors $x^{t-1}_1,\dots,x^{t-1}_N$ are proxies for the latents $u^{t-1}_1, \dots, u^{t-1}_M$ then, for each donor $x_i$, there exists a unique function $h_i$ such that for all time points $t$ we have: 
   \begin{equation}
   \label{eq:appendix donor forecast}
   \E\left(x^t_i \right) = \E\left(h_i(x^{t-1}_1, \dots, x^{t-1}_N, P(\epsilon^t_{x_i},\epsilon^t_u)) \right)
   \end{equation}
  \end{thm*}

\begin{proof}
Define $x^{t-1}:=x^{t-1}_1,\dots,x^{t-1}_N$, $u^{t-1}:=u^{t-1}_1, \dots, u^{t-1}_M$, and $\epsilon^t_i:=(\epsilon^t_{x_i}, \epsilon^t_u)$ for ease of exposition. First, note that we can write the causal mechanism for $x^t_i$ as 
    \begin{equation}\label{equation: causal mechanism}
    P^t(x^t_i \mid u^t, \epsilon_{x_i}^t)=P^t(x^t_i \mid u^{t-1}, \epsilon_u^t, \epsilon_{x_i}^t)=P^t(x_i^t \mid u^{t-1}, \epsilon^t_i)
    \end{equation}
where the first equality follows from the fact $u^t$ is a deterministic function of $u^{t-1}$, and $\epsilon_u^t$. Now, using the proxy variable completeness condition (Def.~\ref{definition: completeness}), we can relate the causal mechanism for $x_i^t$ to $x^{t-1}$ via a function\footnote{A proof for the existence of $H_i^t$ can be found in section F of the supplementary material in \citealt{shi2023theory}, as well as appendix A of \citealt{zeitler2023non}. A simple example of such a function is: $P(Y) = \int P(Y \mid X) P(X)dx = \int H(Y, X) P(X)dx.$} $H^t_i$ as
    \begin{equation}\label{equation: existence of H}
        \begin{aligned}
        P^t(x_i^t \mid u^{t-1}, \epsilon^t_i) &=
        \int H_i^t(x_i^t, x^{t-1}, \epsilon^t_i) \underbrace{P^t(x^{t-1} \mid u^{t-1}, \epsilon^t_i)}_{P^t(x^{t-1} \mid u^{t-1})}dx^{t-1} 
        \end{aligned}
    \end{equation}
This implies that
 \begin{equation*}
        \begin{aligned}
       \E\left(x_i^t \mid u^{t-1}\right) &= \iiint x_i^t\, H_i^t(x_i^t,x^{t-1}, \epsilon^t_i) P^t(x^{t-1} \mid u^{t-1})P(\epsilon^t_i)\,d\epsilon^t_i\,dx_i^t\,dx^{t-1}  \\
       &= \int{P^t(x^{t-1} \mid u^{t-1})}\int \underbrace{\left[\int{x_i^t\,H_i^t(x_i^t, x^{t-1}, \epsilon^t_i)}dx_i^t\right]}_{g_i^t(x^{t-1}, \epsilon^t_i)}P(\epsilon^t_i)\,d\epsilon^t_i\,dx^{t-1}\\
       &=\int \underbrace{\E_{P(\epsilon^t_i)}\left(g_i^t(x^{t-1}, \epsilon^t_i)\right)}_{h_i^t(x^{t-1}, P(\epsilon^t_i))}P^t(x^{t-1} \mid u^{t-1})\,dx^{t-1} \\
       &=\E\left(h_i^t(x^{t-1}, P(\epsilon^t_i)) \mid u^{t-1} \right)
        \end{aligned}
    \end{equation*}
Marginalising over $u^{t-1}$ yields the result that
\begin{equation*}
       \E\left(x_i^t\right) =
       \E\left(h_i^t(x^{t-1}, P(\epsilon^t_i))\right)
\end{equation*}

Next, we prove that $h^t_i$ doesn't depend on time, by showing that the solution to the integral equation in Eq.~\ref{equation: existence of H} for time point $t$ is also a solution for any other time point $t'$. 
Remember that causal mechanisms are invariant and so don't depend on $t$. Consider the left hand side of Eq.~\ref{equation: existence of H}:
   $$
        P^t(x_i^t \mid u^{t-1}, \epsilon^t_i) =
        \int H_i^t(x_i^t, x^{t-1}, \epsilon^t_i) P^t(x^{t-1} \mid u^{t-1})\,dx^{t-1}. 
    $$
By Eq.~\ref{equation: causal mechanism}, $P^t(x_i^t \mid u^{t-1}, \epsilon^t_i)$ is a causal mechanism and so doesn't depend on $t$. Now consider the $P^t(x^{t-1} \mid u^{t-1})$ term under the integral on the right hand side. This distribution Markov factorises according to the structure of the DAG in Figure~\ref{fig:syn_cont}, resulting in products of causal mechanisms. Hence, a solution to the integral equation for one time point $t$ is a solution for any other time point. 
    
Finally, we prove uniqueness of $h_i^t$ for a given time point, as this implies there exists a unique function $h_i$ for all time points. Suppose that $h_i^t$ and $\widetilde{h}_i^t$ are both solutions at time $t$, such that
$$\E\left(h_i^t(x^{t-1}, P(\epsilon^t_i)) \mid u^{t-1}\right) = \E\left(\widetilde{h}_i^t(x^{t-1}, P(\epsilon^t_i)) \mid u^{t-1} \right)$$
Therefore, $\E\left(h_i^t(x^{t-1}, P(\epsilon^t_i)) - \widetilde{h}_i^t(x^{t-1}, P(\epsilon^t_i)) \mid u^{t-1} \right) = 0$, and the proxy variable completeness condition (Def.~\ref{definition: completeness}) implies that 
$h_i^t(x^{t-1}, P(\epsilon^t_i)) = \widetilde{h}_i^t(x^{t-1}, P(\epsilon^t_i))$, completing the proof.
%\begin{rmk}
\end{proof}

The result can be generalised to the case where the evolution depends on $T$ time points by modifying the SCSCM such that
\begin{enumerate}
   \item $u_i^t=m_i^t(u_i^{t-1},\dots,u_i^{t-1-T}, \epsilon^t_{u_i})$ 
   \item $x^t_i=f_i^t(u_1^t,\dots,u_M^t,\dots,u_1^{t-T},\dots,u_M^{t-T}, \epsilon^t_{x_i})$ 
   \item $y^t=g^t(u_1^t,\dots,u_M^t,\dots,u_1^{t-T},\dots,u_M^{t-T}, I^t, \epsilon^t_{y})$ 
\end{enumerate}
and updating the proxy variable completeness condition (Def.~\ref{definition: completeness}) accordingly. 
%\end{rmk}

%\newpage
\begin{figure}[t]
    \centering
    \includegraphics[scale=0.6]{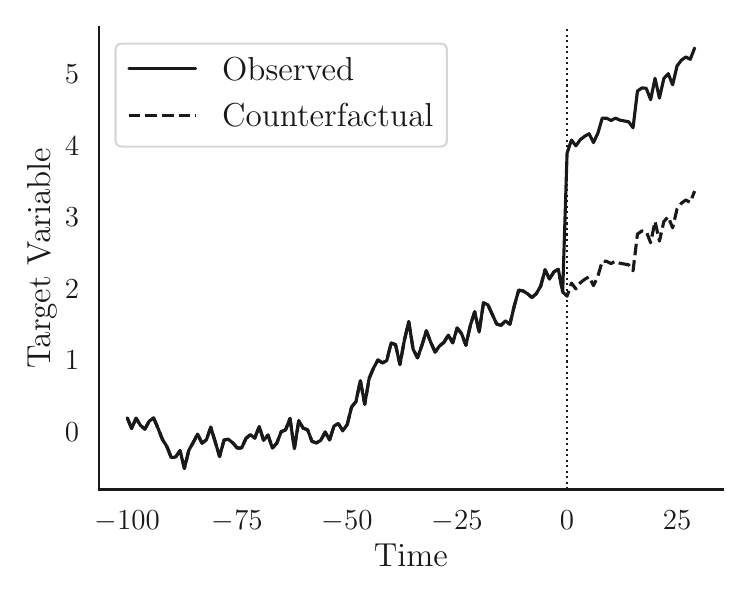}
    \caption{Example of a target variable $y$ from a simulated dataset. The intervention happens at time 0, increasing $y$ by $\tau=2$.}
    \label{fig:y_sim}
\end{figure}

\section{Simulated Data}\label{section: sim data}
We construct simulated datasets in Section~\ref{subsection: simulated data} according to the following data generating process \citep[similar to the local linear trend model with long-term slope from][]{brodersen2015inferring}.
As mentioned in the main text, we use the Einstein summation convention such that a repeated index implies summation, e.g.\ $\beta_{ij}\,u_j^t := \sum_j \beta_{ij}\,u_j^t$.
\begin{align*}
u_i^{t+1} &\sim \mathcal{N}\left(u_i^t + \delta_i^t\,,\,\sigma_u\right) \\
\delta_i^{t+1} &\sim \mathcal{N}\left(S_i + \rho_i\left(\delta_i^t-S_i\right),\,\sigma_\delta\right) \\
y^t &\sim \mathcal{N}\left(\alpha_i\,u_i^t+\tau I^t\,,\,\sigma_y\right) \\
x_i^t &\sim \mathcal{N}\left(\beta_{ij}\,u_j^t+\tau_{x_i} I^t\,,\,\sigma_x\right)
\end{align*}
We generate 2000 datasets for each panel in Figure~\ref{fig:sims_low_med_high_noise}. A dataset consists of 1 target $y$, 10 latents $u_i$, and 1000 potential donors $x_i$ (out of which we select 10 donors for constructing SCs). The time-series data have 100 time points pre-intervention, and 30 time points post-intervention. The causal effect of the intervention on the target is $\tau=2$, and 80\% of the potential donors are invalid, with spillover effects $\tau_{x_i}=-2$ (the remaining 20\% are valid, with $\tau_{x_i}=0$).
In the above, $I^t$ is an indicator variable for the intervention that is 0 pre-intervention and 1 post-intervention. 
The long-term slope is sampled as $S_i\sim\mathcal{N}(0.1, 0.1)$, with $\rho_i\sim \mathcal{U}\left(0, 1\right)$ interpolating between this slope and a random walk. We set $\sigma_u=1$, $\sigma_\delta=\sigma_y=0.1$, and 
$\sigma_x\in\{0.1, 0.5, 1.0\}$ for the low, medium, and high donor noise levels (the high noise level is comparable in magnitude to $\tau$).
For simplicity, we set the coefficients as $\alpha_i=1$, and $\beta_{ij}=1$ (except for the datasets demonstrating the latent shift in Figure~\ref{fig:u_shift}, where we set $\beta_{i1}=0$ for a random subset of valid donors).
Figure~\ref{fig:y_sim} provides an illustrated example of a target variable generated by the process described above.

\section{Sparse Donor Weights}\label{section: sparse weights}

\begin{figure}[t]
    \centering
    \includegraphics[scale=0.7]{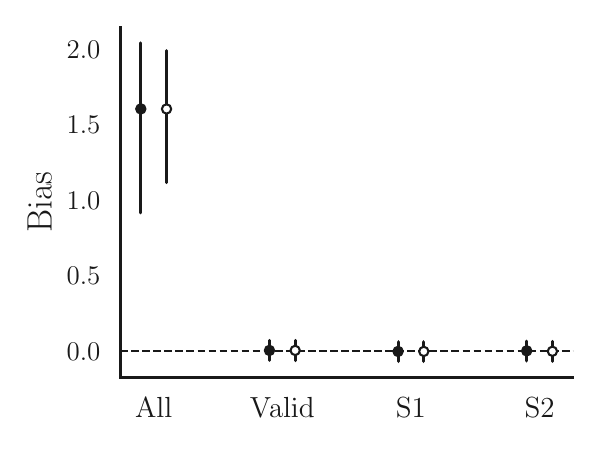}
    \caption{Low noise case of Figure~\ref{fig:sims_low_med_high_noise}. The filled circles show the results discussed in the main text, where each SC model is estimated using 10 donors sampled from the set of PVDs. The open circles show the alternative sparse regularization approach, where each SC model is estimated using the full set of PVDs, and we employ sparsity-inducing priors (Eq.~\ref{eq:sparse_prior}) to control the number of non-zero weights. The expected bias is identical for both approaches. The variance in the \emph{All} case is slightly lower. This is because the number of donors with non-zero weight typically ends up being a bit larger than 10, and so the fraction of invalid donors is more concentrated around 80\%.}
    \label{fig:bias_compare}
\end{figure}

As described in Section~\ref{subsection: simulated data}, for each of the 2000 simulated datasets, we construct SCs by sampling 10 donors from the set of potentially valid donors (PVDs). The set of PVDs is a subset of the pool of 1000 potential donors, and depends on the specific selection procedure used to identify valid donors (ie. \emph{All}, \emph{Valid}, \emph{S1}, or \emph{S2}). 
As an alternative to sampling 10 donors, we can construct SCs using the full set of PVDs, and employ regularization to select approximately 10 donors to have non-zero weights based on the pre-treatment fit \citep[e.g.,][]{brodersen2015inferring,ben-michael_2021}.
With this sparse regularization approach, the total number of parameters is equal to the number of PVDs identified by the selection procedure. 
For the simulations described in Appendix~\ref{section: sim data}, SC models estimated using the full set of PVDs have 1000 parameters in the \emph{All} case, 200 parameters in the \emph{Valid} case, 10 parameters in the \emph{S1} case, and approximately 200 parameters in the \emph{S2} case (but the exact number varies across datasets).

To enforce sparsity in the donor weights, we set the prior to be a discrete mixture 
of normal distributions \citep{betancourt_modelling_sparsity_2021}.
\begin{equation}
\label{eq:sparse_prior}
\beta_i\sim\eta\,\mathcal{N}\left(0, \sigma_1\right) + \left(1-\eta\right)\mathcal{N}\left(0,\sigma_2\right),\qquad 0<\eta<1,\qquad\sigma_1 \ll \sigma_2
\end{equation}
The parameter $\eta$ controls how weights tend to cluster close to zero via the narrow distribution, and can be interpreted as the expected fraction of donors with effectively zero weight.
This is similar to the spike and slab prior of \cite{mitchell_beauchamp_1988}. However, replacing the exact zero values with a narrow distribution of ``irrelevant'' values facilitates straightforward estimation with standard probabilistic programming languages. 

In Figure~\ref{fig:bias_compare}, we compare this sparse regularization approach to the sampling approach from the main text. For this data generating process, the expected bias is identical. The variance in the \emph{All} case is slightly lower. This is because the number of donors with non-zero weight typically ends up being a bit larger than 10, and so the fraction of invalid donors is more concentrated around 80\%.

\end{document}